\documentclass{eptcs}

\usepackage[utf8]{inputenc}
\usepackage[T1]{fontenc}    
\usepackage[english]{babel}
\usepackage{algorithm}
\usepackage[fleqn]{amsmath}
\usepackage{amsthm}
\usepackage{stmaryrd}
\usepackage[capitalize,nameinlink,noabbrev]{cleveref} 

\usepackage{todonotes}
\usepackage{enumerate}
\usepackage{amssymb}
\usepackage{tikz}
\usetikzlibrary{shapes,shapes.geometric,positioning,arrows,automata,calc,fit}
\tikzset{every state/.style={inner sep=2pt},elliptic state/.style={draw,ellipse}}

\newcommand{\qs}[2]{#1_#2} 
\renewcommand{\downarrow}{2}
\renewcommand{\uparrow}{1}

\theoremstyle{definition}
\newtheorem{thm}{Theorem}[section]
\newtheorem{corollary}[thm]{Corollary}
\newtheorem{dfn}[thm]{Definition}
\newtheorem{lem}[thm]{Lemma}

\crefname{enumi}{Case}{Cases}
\crefname{thm}{Theorem}{Theorems}
\crefname{corollary}{Corollary}{Corollaries}
\crefname{lem}{Lemma}{Lemmas}
\crefname{dfn}{Definition}{Definitions}
\crefname{cnj}{Conjecture}{Conjectures}

\title{Distribution of Behaviour\\ into \\ Parallel Communicating Subsystems}
\author{Omar al Duhaiby \qquad\qquad Jan Friso Groote
\institute{Eindhoven University of Technology\\ Eindhoven, The Netherlands}
\email{o.z.alzuhaibi@tue.nl \quad\qquad\ j.f.groote@tue.nl}
}
\def\titlerunning{Parallel Communicating Decompositions}
\def\authorrunning{O. al Duhaiby \& J.F. Groote}
\date{May 2019}
\begin{document}
\maketitle

\begin{abstract}
The process of decomposing a complex system into simpler subsystems has been of interest to computer scientists over many decades, for instance, for the field of distributed computing. In this paper, motivated by the desire to distribute the process of active automata learning onto multiple subsystems, we study the equivalence between a system and the total behaviour of its decomposition which comprises subsystems with communication between them. We show synchronously- and asynchronously-communicating decompositions that maintain branching bisimilarity, and we prove that there is no decomposition operator that maintains divergence-preserving branching bisimilarity over all LTSs. 
\end{abstract}

\section{Introduction}

The process of decomposing a complex system into simpler subsystems is the cornerstone of behavioural analysis regardless of where it is applied, to the atom or to the human psyche. Studying the relationship between a complex system and the total behaviour of its decomposition is the subject matter of this paper. However, instead of atoms or human brains, in the field of formal methods, we simply dissect automata. This paper studies how the behaviour of a Labelled Transition System (LTS) can be distributed into a parallel (de)composition of communicating subsystems while maintaining behavioural equivalence.

\paragraph{Motivation}
This work was motivated by a case study in the industry~\cite{duhaiby2018Pitfalls} based on which we pursued the possibility of deducing the internal components of a system based on the model inferred by the active model learning technique~\cite{angluin_learning_1987}. If it were possible at all, then the learned system must be equivalent to the parallel decomposition deduced. 

Our goal is to examine the possibility of such deduction and we do that by devising a decompositioning scheme with certain assumptions, then examine its equivalence with the original system.

\paragraph{Related Work} 
Previous work done on the topic of decomposition focuses on uniqueness properties and on producing simpler components. For example the primary decomposition theorem by Krohn and Rhodes states that any automaton can be decomposed into a cascaded product of simpler automata such that the automaton is homomorphic to its decomposition~\cite{krohn1965algebraic}. And in 1998, Milner and Moller introduced a semantics of parallel decompositions comprising non-communicating subsystems~\cite{Milner1993Unique}, and they proved that any finite system of behaviour can be decomposed into a unique set of prime parallel non-communicating subsystems. While Milner and Moller's theorem was in strong bisimulation set in the Calculus of Communicating Systems (CCS), Luttik~\cite{luttik2015Unique} later introduced a more general framework of commutative monoids that can be used to establish unique decomposition in weak and branching bisimulation semantics. 

In contrast to those works, we consider decomposing a system based on partitioning its alphabet into disjoint sets and we require that the subsystems communicate. In that sense, this work is more similar to Brinksma et al.~\cite{brinksma1995} and Hultstr{\"{o}}m~\cite{hultstrom1995structural} who made a decomposition based on a given partition of actions. Another similarity is the need for synchronisation between subsystems as defined in~\cref{sec:decomp-sync-def}. 

\paragraph{Contribution} We define two decompositions of parallel communicating subsystems, one synchronous and the other asynchronous, and we prove that both decompositions maintain branching bisimilarity~\cite{vanGlabbeek1996Branching} with the source automaton. We also prove that there is no way of decomposing an automaton (under certain conditions) such that it is divergent-preserving branching-bisimilar~\cite{vanGlabbek2008BranchingExplicitDivergence} to the resulting decomposition.
We consider branching bisimilarity, but the results of this paper easily carry over to other equivalences such as weak bisimilarity.

\paragraph{Outline} The outline of this paper is as follows. \cref{sec:preliminaries} introduces the preliminaries. \cref{sec:decomp-general} defines and discussed the general decomposition operator on which we base our arguments. \cref{sec:decomps-both} defines two decompositions of communicating subsystems, one for synchronous communication and the other for asynchronous communication, and proves that each maintains a branching bisimulation relation with the source automaton. \cref{sec:proofdpbb} contains the proof that there is no way of decomposing an automaton, through our general decomposition operator, such that it maintains divergence preserving branching bisimulation with its decomposition. Finally, \cref{sec:conclusion} discusses the results and interprets them in light of our initial motivation.

\paragraph{Acknowledgement}
We wish to thank Rick Erkens, Bas Luttik, Thomas Neele, Joshua Moerman, Pieter Cuijpers, Bharat Garhewal, Hans van Wezep, Arjan Mooij and the referees of the EXPRESS/SOS 2019 workshop for their thorough feedback, for sharing their knowledge, and for their motivation and support.

\section{Preliminaries}
\label{sec:preliminaries}

In this section, we present the preliminaries of labelled transition systems, the synchronous product and bisimulation relations, aided by~\cite{groote2004book}. We start with the definition of a labelled transition system (LTS).

\begin{dfn} [LTS]
We define our LTS as a four-tuple $(S,\Sigma, \rightarrow, s_0)$ where:
\begin{itemize}
\item $S$ is a non-empty finite set of states.
\item $\Sigma$ is the alphabet, also referred to as the action set.
\item ${\rightarrow} \subseteq S \times \Sigma \times S$ is a transition relation.
\item $s_0$ is the initial state.
\end{itemize}
\end{dfn}

We use the notation $x \xrightarrow{a} y$ to express a transition with action $a$ from state $x$ to state $y$. This and variations of it are formally defined as follows.

\begin{dfn} [Transition Relation]
Let $(S,\Sigma, \rightarrow, s_0)$ be an LTS with $s, s' \in S$ and $a \in \Sigma \cup \{\tau\}$, where $\tau$ is the internal/unobservable action. We use the following notations:

\begin{tabular}{ll}
$s \xrightarrow{a} s'$ & iff $\langle s, a, s' \rangle \in {\rightarrow}$. \\
$s \xrightarrow{a} $ & iff there is an $s'$ such that $s \xrightarrow{a} s'$. \\
$s \not\xrightarrow{a} $ & iff there is no $s'$ such that $s \xrightarrow{a} s'$. \\
$s \xrightarrow{a}^{*} s_n $ & iff there are $s_1, s_2, \ldots, s_n \in S$ such that $s \xrightarrow{a} s_1 \xrightarrow{a} s_2 \xrightarrow{a} \cdots \xrightarrow{a} s_n$. \\
$s \xrightarrow{a}^{\omega}$ & iff there are $s_1, s_2, \ldots \in S$ such that $s \xrightarrow{a} s_1$ and for all $i \in \mathbb{N}$, $s_i \xrightarrow{a} s_{i+1}$.
\end{tabular}
\label{dfn:transitions}
\end{dfn}

Next, we define complementary actions, i.e., actions on which communicating systems synchronise. Then we define the synchronous product of two automata, and show what role complementary actions play in computing it. 

\begin{dfn}[Co-actions]
For an arbitrary action $a$ that is not $\tau$, the action $\overline{a}$ (read as \emph{a bar}) is called its co-action. Also, $\overline{(\overline{a})} = a$. We say that actions $a$ and $\overline{a}$ are \emph{complementary} to each other and we call them a pair of \emph{complementary actions}.

We lift this operator to sets of actions such that $\overline{\Sigma} = \{\overline{a} \mid a \in \Sigma\}$.
\end{dfn}

\begin{dfn} [Synchronous Product]
The synchronous product of two LTSs \\
$(S_1,\Sigma_1, \rightarrow_1, q_0) \times$ 
$(S_2,\Sigma_2, \rightarrow_2, r_0)$ is the tuple $(S_1 \times S_2,\Sigma_x, \rightarrow_x,$ $(q_0, r_0))$ \\
where $\Sigma_x = (\Sigma_1 \cup \Sigma_2) \setminus \{a,\overline{a} \mid a \in \Sigma_1 \land \overline{a} \in \Sigma_2 \}$.\\
The transition relation ${\rightarrow_x} \subseteq (S_1 \times S_2) \times \Sigma_x \times (S_1 \times S_2)$ is defined as follows:
\[ 
\begin{cases} 
      (s,t) \xrightarrow{a} (s',t) & \text{iff } a \in \Sigma_1 \land \overline{a} \not\in \Sigma_2 \land s \xrightarrow{a}_1 s', \\
      (s,t) \xrightarrow{a} (s,t') & \text{iff } a \in \Sigma_2 \land  \overline{a} \not\in \Sigma_1  \land t \xrightarrow{a}_2 t', \text{ and} \\
      (s,t) \xrightarrow{\tau} (s',t') & \text{iff } a \in \Sigma_1 \land \overline{a} \in \Sigma_2 \land s \xrightarrow{a}_1 s' \land t \xrightarrow{\overline{a}}_2 t',
   \end{cases}
\]
where $\tau$ is the unobservable action.
\label{dfn:syncproduct}
\end{dfn}

Next, we define two notions of behavioural equivalence.

\begin{dfn}[Branching bisimulation]
Given an LTS $(S, \Sigma, \rightarrow, s_0)$ and a relation $\mathcal{R} \subseteq S \times S$. We call $\mathcal{R}$ a branching bisimulation relation iff for all states $s,t \in S$ such that $\langle s,t \rangle \in \mathcal{R}$, it holds that:
\begin{enumerate}
\item if $s \xrightarrow{a} s'$, then:
\begin{itemize}
\item $a=\tau$ and $\langle s', t \rangle \in \mathcal{R}$; or
\item $t \xrightarrow{\tau}^{*} t' \xrightarrow{a} t''$, $\langle s, t' \rangle \in \mathcal{R}$ and $\langle s', t'' \rangle \in \mathcal{R}$.
\end{itemize}
\item Symmetrically, if $t \xrightarrow{a} t'$, then:
\begin{itemize}
\item $a=\tau$ and $\langle s,t' \rangle \in \mathcal{R}$; or
\item $s \xrightarrow{\tau}^{*} s' \xrightarrow{a} s''$, $\langle s',t \rangle \in \mathcal{R}$ and $\langle s'',t' \rangle \in \mathcal{R}$.
\end{itemize}
\end{enumerate}

Two states $s$ and $t$ are branching \textit{bisimilar}, denoted $s \leftrightarroweq_{b} t$ iff there is a branching bisimulation relation $\mathcal{R}$ such that $\langle s,t \rangle \in \mathcal{R}$. Two LTSs $P$ and $Q$ are branching bisimilar, denoted $P \leftrightarroweq_{b} Q$, iff their initial states are.

\label{dfn:branchingbisim}
\end{dfn}

A state $s$ with $s \xrightarrow{\tau}^{\omega}$ is called \emph{divergent}. Hence, a state with a $\tau$ loop is also called divergent. 
Branching bisimulation does not preserve divergence, i.e., a divergent state can be branching bisimilar to a non-divergent one. Therefore, a stronger equivalence relation, namely divergence-preserving branching bisimulation, is defined below.

\begin{dfn}[Divergence-preserving branching bisimulation]
Given an LTS $(S, \Sigma, \rightarrow, s_0)$ and a relation $\mathcal{R} \subseteq S \times S$. We call $\mathcal{R}$ a divergence-preserving branching bisimulation relation iff it is a branching bisimulation relation and for all states $s,t \in S$ with $\langle s,t \rangle \in \mathcal{R}$, there is an infinite sequence $s \xrightarrow{\tau} s_1 \xrightarrow{\tau} s_2 \xrightarrow{\tau}$ with $\langle s_i, t \rangle \in \mathcal{R}$ for all $i > 0$ iff there is an infinite sequence $t \xrightarrow{\tau} t_1 \xrightarrow{\tau} t_2 \xrightarrow{\tau}$ and $\langle s, t_i \rangle \in \mathcal{R}$ for all $i > 0$.

Two states $s$ and $t$ are divergence-preserving branching bisimilar, denoted $s \leftrightarroweq_{db} t$ iff there is a divergence-preserving branching bisimulation relation $\mathcal{R}$ such that $\langle s,t \rangle \in \mathcal{R}$. Two LTSs $P$ and $Q$ are divergence-preserving branching bisimilar, denoted $P \leftrightarroweq_{db} Q$, iff their initial states are.
\label{dfn:dpbb}
\end{dfn}

\section{The Decomposition Operation}
\label{sec:decomp-general}
In this section, we establish a decompositioning scheme that is based on action partitioning in order to refer to it as \textit{the} general decomposition operation on which our theorems apply.
Our general decomposition operation is a function transforming a single LTS, given two disjoint actions sets, into two LTSs. It is formally defined as follows.

\begin{dfn} [General Decomposition Operation]
Given an LTS $M$ with alphabet $\Sigma$ and given two alphabets $\Sigma_1, \Sigma_2$ such that $\Sigma = \Sigma_1 \cup \Sigma_2$ and $\Sigma_1~\cap~\Sigma_2 = \emptyset$, we call $G$ a \emph{general decomposition operator} iff $G(M, \Sigma_1, \Sigma_2) = (M_1, M_2)$ such that $M_1$ has alphabet $\Sigma_{M_1}$ with ${\Sigma_1 \subseteq \Sigma_{M_1}}$ and ${\Sigma_{M_1} \cap \Sigma_2 = \emptyset}$, and likewise, $M_2$ has alphabet $\Sigma_{M_2}$ with $\Sigma_2 \subseteq \Sigma_{M_2}$ and $\Sigma_{M_2} \cap \Sigma_1 = \emptyset$.
\label{dfn:decomposition_general}
\end{dfn}

We refer to a method of decomposing automata as a \emph{decomposition operation} whereas the result of such transformation is called a \emph{decomposition}. A decomposition comprises two or more automata. This transformation is depicted in \cref{fig:decompExample}. Throughout the paper, we compare LTSs to the synchronous product of the decomposition, and if a a certain bisimulation relation holds between these two, then we say that the operation \emph{maintains} that relation.

The idea behind this transformation is, given a partition of actions of a system, to generate two subsystems, each using exclusively one of the two parts. 

\paragraph{Recursive decomposition.}
Note that \cref{dfn:decomposition_general} can easily be applied recursively allowing to decompose behaviour into multiple components.
As the alphabets over which an automaton can be decomposed can be empty, the number of components over which behaviour can be split can even be arbitrarily large.

\section{Branching Bisimilar Decompositions}
\label{sec:decomps-both}

In this section, we define two decomposition operations that are designed to maintain branching bisimilarity, and we actually prove that they do. The first one ($decomp_s$) decomposes into synchronously communicating subsystems while the second ($decomp_a$) decomposes into asynchronously communicating ones. In both of these operations, we build two automata that pass control between one another using $c$ messages. Only if an automaton seizes control can it perform one of its actions, otherwise it has to wait for the other to hand control over. Formal definitions and more detail follow.

\subsection{Decomposing into Synchronous Subsystems}
\label{sec:decomp-sync-def}

We define the decomposition of synchronous subsystems, summarised in \cref{fig:decompRules} in two patterns; the top dictates the decomposition of every state in the source LTS while the bottom dictates the decomposition of every transition. An omitted third pattern is symmetric to the second such that the transition's label simply belongs to the second subsystem rather than the first.

\begin{figure}
    \centering
\begin{tikzpicture}[->,shorten >=1pt,auto,on grid,semithick, inner sep=0.3mm,node distance=2cm, minimum size=0pt,bend angle=20,  transform shape]
\node [state,inner sep=1mm, minimum size=0pt] (r) {$r$};
\node [state,inner sep=0.3mm, minimum size=0pt,draw=none] (decomp) [right= of r] {$\xrightarrow{decomp_s}$};
\node[state,inner sep=1mm, minimum size=0pt]         (s') [below right= 2cm and 0.75cm of r] {$s$};
\node[state,inner sep=1mm, minimum size=0pt]         (s) [below left= 2cm and 0.75cm of r] {$r$};
\node [state,inner sep=0.3mm, minimum size=0pt,draw=none] (decomp2) [right= of s'] {$\xrightarrow{decomp_s}_{a \in \Sigma_1}$};
\path [every node/.style={font=\footnotesize}, bend angle=37]
(s) edge              node {$a$} (s')
;




\node[state,inner sep=0.6mm, minimum size=0pt] (r1)                   [right =of decomp] {$r_\uparrow$};
\node[state,inner sep=0.6mm, minimum size=0pt]         (r2) [right=of r1] {$r_\downarrow$};
\node[state,inner sep=0.6mm, minimum size=0pt]         (s1) [right =of decomp2] {$r_\uparrow$};
\node[state, elliptic state, draw=none,minimum size=0pt]         (tr1) [right= 1.5cm of s1] {$t_{a,s}$};
\node[state,inner sep=0.6mm, minimum size=0pt]         (s2) [right= 1.5cm of tr1] {$s_\uparrow$};

\path [every node/.style={font=\footnotesize}, bend angle=17]
(r1) edge  [bend left] node {$c_{r_\uparrow,r_\downarrow}$} (r2)
(r2) edge  [bend left] node {$\overline{c_{r_\downarrow,r_\uparrow}}$} (r1)
(s1) edge   node {$a$} (tr1)
(tr1) edge node {$t_{s}$} (s2)
;

\node[] (par) [right= 1cm of r2] {\Huge{$\mid$}};
\node[] (par2) [right= 1cm of s2] {\Huge{$\mid$}};
\node[state,inner sep=0.6mm, minimum size=0pt] (r1_2)                   [right =of r2] {$r_\uparrow$};
\node[state,inner sep=0.6mm, minimum size=0pt]         (r2_2) [right=of r1_2] {$r_\downarrow$};
\node[state,inner sep=0.6mm, minimum size=0pt]         (s1_2) [right =of s2] {$r_\uparrow$};
\node[state,inner sep=0.6mm, minimum size=0pt]         (s2_2) [right=of s1_2] {$s_\uparrow$};

\path [every node/.style={font=\footnotesize}, bend angle=17]
(r1_2) edge  [bend left] node {$\overline{c_{r_\uparrow,r_\downarrow}}$} (r2_2)
(r2_2) edge  [bend left] node {$c_{r_\downarrow,r_\uparrow}$} (r1_2)
(s1_2) edge  node {$\overline{t_s}$} (s2_2)
;

\end{tikzpicture}
    \caption{The two patterns that delineate the operator $decomp_s$ (\cref{dfn:decomposition}).}
    \label{fig:decompRules}
\end{figure}
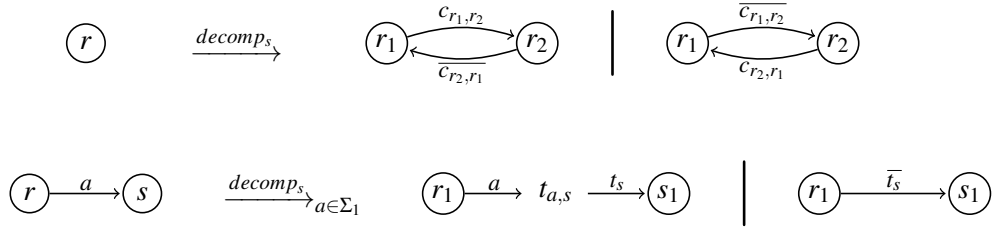

\begin{dfn} [Synchronous Decomposition Operation]
Given an LTS $M = (S, \Sigma, \rightarrow, q)$ and two alphabets $\Sigma_1, \Sigma_2$ such that $\Sigma = \Sigma_1 \cup \Sigma_2$ and $\Sigma_1 \cap \Sigma_2 = \emptyset$, then we can decompose $M$ over $\Sigma_1$ and $\Sigma_2$ by applying the following operation:

$decomp_s (M, \Sigma_1, \Sigma_2) = (M_1, M_2)$ where:
\begin{enumerate}
\item $M_1 = (S_C \cup S_{T_1}, \Sigma_1 \cup \Sigma_{S_1}, \rightarrow_1, (q,1))$\\
$M_2 = (S_C \cup S_{T_2}, \Sigma_2 \cup \Sigma_{S_2}, \rightarrow_2, (q,1))$. 
\item For every state $s$ in $S$, we introduce two states $(s,1), (s,2) \in S_C$: \begin{equation}
  \begin{aligned}
&S_{C_1} = \{(s,\uparrow) \mid s \in S\} 
\hspace{1cm} S_{C_2} = \{(s,\downarrow) \mid s \in S\}
\hspace{1cm} S_C = S_{C_1} \cup S_{C_2} &\\
  \end{aligned}
\end{equation}
\paragraph{Notation.} Tuple-states of the form $(s,i)$ such as $(s,1)$ and $(s,2)$ are shortened to $s_i$. Therefore, it is to be held throughout the paper that $s_i$ is derived from $s$ rather than it being a completely unrelated symbol to $s$.
\item The set of $c$-actions is defined as follows:
\begin{equation}
  \begin{aligned}
&\Sigma_{C} = \{c_{s_\uparrow,s_\downarrow}, \overline{c_{s_\downarrow,s_\uparrow}} \mid s_\uparrow \in S_{C_1}, s_\downarrow \in S_{C_2}\}\\
  \end{aligned}
\end{equation}
\item The sets of $t$ actions and $t$ states are defined as follows:
\begin{equation}
  \begin{aligned}
&\Sigma_{T_1} = \{t_{s_\uparrow} \mid s_\uparrow \in S_C \} 
&\Sigma_{T_2} = \{t_{s_\downarrow} \mid s_\downarrow \in S_C \}\\
&S_{T_1} = \{t_{a,s_\uparrow} \mid a \in \Sigma_1, s_\uparrow \in S_{C_1}\} 
&S_{T_2} = \{t_{a,s_\downarrow} \mid a \in \Sigma_2, s_\downarrow \in S_{C_2}\} \\
  \end{aligned}
\end{equation}
\item The complete sets of actions of $M_1$ and $M_2$ are respectively defined as: 
\begin{equation}
  \begin{aligned}
&\Sigma_{S_1} = \Sigma_{T_1} \cup \Sigma_{C} \cup \overline{\Sigma_{T_2}} \hspace*{25pt}
\Sigma_{S_2} = \Sigma_{T_2} \cup \overline{\Sigma_{C}} \cup \overline{\Sigma_{T_1}}\\
  \end{aligned}
\end{equation}

\item The transition relations
${\rightarrow_i} \subseteq (S_C \cup S_{T_i}) \times (\Sigma_i \cup \Sigma_{S_i}) \times (S_C \cup S_{T_i})$ are defined as follows. For $i,j \in \{\uparrow,\downarrow\}$ and $i \neq j$, $\rightarrow_i$ is the minimal relation satisfying the following:

\begin{enumerate}
\item For all $s \in S$ and for all $c_{\qs{s}{i},\qs{s}{j}} \in \Sigma_C$:
\begin{equation}
	\qs{s}{i} \xrightarrow{c_{\qs{s}{i},\qs{s}{j}}}_i \qs{s}{j} \\
    \qs{s}{i} \xrightarrow{\overline{c_{\qs{s}{i},\qs{s}{j}}}}_j \qs{s}{j} \\
    \label{eq2}
\end{equation}

\item For all $s,s' \in S$, and all $a \in \Sigma_i$, if $s \xrightarrow{a} s'$, then: 
\begin{equation}
  \begin{aligned}
     \qs{s}{i} \xrightarrow{a}_i t_{a,\qs{s'}{i}} \xrightarrow{t_{\qs{s'}{i}}}_i \qs{s'}{i}\\
     \qs{s}{i} \xrightarrow{\overline{t_{\qs{s'}{i}}}}_j \qs{s'}{i}\\
     \label{eq1}
  \end{aligned}
\end{equation}

\end{enumerate}
\end{enumerate}

\label{dfn:decomposition}
\end{dfn}

Two classes of actions are introduced, $c$-actions and $t$-actions. The $c$-actions come in pairs, and they resemble passing a control token between $M_1$ and $M_2$. For instance, looking at \cref{fig:decompExample}, when, at some state $r \in S$ for which a pair of states $r_\uparrow, r_\downarrow \in S_C$ exists in both $M_1$ and $M_2$, and control is to be passed from $M_1$ to $M_2$, then a pair of complementary $c$ actions synchronises, namely, actions $c_{r_\uparrow,r_\downarrow}$ and $\overline{c_{r_\uparrow,r_\downarrow}}$, to produce a synchronous transition in both machines from $r_\uparrow$ to $r_\downarrow$. Likewise, actions $c_{r_\downarrow,r_\uparrow}$ and $\overline{c_{r_\downarrow,r_\uparrow}}$ synchronise to pass control in the opposite direction from $M_2$ to $M_1$. 

The $t$-actions are introduced to synchronise transitions occurring in one machine with the other. In addition, they require the introduction of $t$-states. Observe \cref{fig:decompExample} where an $a_1$ transition occurs in $M_1$. The aim is the transition $r_\uparrow \xrightarrow{a_1} s_\uparrow$, but in order to synchronise this with $M_2$, we introduce a middle state $t_{a_1,s_\uparrow} \in S_{T_1}$ from which the only possible transition is $t_{a_1,s_\uparrow} \xrightarrow{t_{s_\uparrow}} s_\uparrow$ which synchronises with the transition $r_\uparrow \xrightarrow{\overline{t_{s_\uparrow}}} s_\uparrow$ in $M_2$. 


The operation ($decomp_s$) can be summarised by two patterns shown in \cref{fig:decompRules}; the top pattern applies to each state and the bottom one applies to each transition.

\paragraph{Is the Decomposition a Simplification?}
The decomposition is obviously larger than the original system. That is due to the nature of an alphabet-partitioning-based decomposition where subsystems must hand control over between one another. Thus, every subsystem must have, for each state in the original, multiple ones expressing where control lies. In special cases, a smaller component may suffice; for example, in \cref{fig:decompExample}, because state $r$ enables no $b$ actions, a state $r_2$ is not needed and all transitions going to $r_2$ can instead go to $r_1$. However, we believe that our definition the way it is is more understandable because of its generality and symmetry. The element of simplification lies not in reducing the size of a system but rather in partitioning its alphabet.


\paragraph{Computing the Synchronous Product.}
For a decomposition $(M_1, M_2)$ by \cref{dfn:decomposition}, the synchronous product $M_x = M_1 \times M_2$ is the LTS $(S_x,\Sigma_1 \cup \Sigma_2, \rightarrow_x, (q_1, q_1))$, where:
\begin{equation}
\begin{aligned}
S_x = S_1 \times S_2 &= (S_C \cup S_{T_1}) \times (S_C \cup S_{T_2}) \\
&= (S_C \times S_C) \cup (S_{T_1} \times S_{T_2}) \\
& \quad \, \cup (S_{T_1} \times S_C) \cup (S_C \times S_{T_2}) \\
\end{aligned}
\end{equation}
with $\Sigma_{S_1}, \Sigma_{S_2}, S_{T_1}, S_{T_2}$ being sets introduced by $decomp_s$. The transition relation $\rightarrow_x$ is defined as follows for $i,j \in \{\uparrow,\downarrow\}$ and $i \neq j$:
\begin{enumerate}
\item if $s \xrightarrow{a} s'$ and $a \in \Sigma_i$ then by~(\ref{eq1}) there is a state $t_{a,\qs{s'}{i}} \in S_{T_i}$ and a pair of complementary actions $t_{\qs{s'}{i}}, \overline{t_{\qs{s'}{i}}} \in \Sigma_{S_i}$ such that: 
\begin{equation}
  \begin{aligned}
    (\qs{s}{i},\qs{s}{i}) \xrightarrow{a}_x s_a \xrightarrow{\tau}_x (\qs{s'}{i},\qs{s'}{i}),
    \label{eq3}
  \end{aligned}
\end{equation}
\[
\text{where } s_a = 
\begin{cases}
      (t_{a,\qs{s'}{i}}, \qs{s}{i}) &\text{if } i=\uparrow,\\
      (\qs{s}{i}, t_{a,\qs{s'}{i}}) &\text{if } i=\downarrow.
\end{cases}
\]
\item For all $s \in S$, there exist $c_{\qs{s}{i},\qs{s}{j}},c_{\qs{s}{j},\qs{s}{i}} \in \Sigma_C$ such that, by~(\ref{eq2}), $\qs{s}{i} \xrightarrow{c_{\qs{s}{i},\qs{s}{j}}}_i \qs{s}{j}$ and $\qs{s}{i} \xrightarrow{\overline{c_{\qs{s}{i},\qs{s}{j}}}}_j \qs{s}{j}$, and thus:
\begin{equation}
  \begin{aligned}
	&(\qs{s}{i},\qs{s}{i}) \xrightarrow{\tau}_x (\qs{s}{j},\qs{s}{j})\\
    \label{eq5} \\
  \end{aligned}
\end{equation}
\end{enumerate}

\label{productofdecomposition}

\subsection{Proof that the synchronous decomposition operation maintains branching bisimulation}
\label{sec:decomp-sync-proof}

In this subsection, we show an application of $decomp_s$ (\cref{dfn:decomposition}) to a sample LTS, we demonstrate that $decomp_s$ maintains branching bisimilarity, and then we prove that branching bisimilarity is maintained through any and all applications of $decomp_s$.

\cref{fig:decompExample} shows the LTS at the left side and its decomposition at the right side. The two patterns shown in \cref{fig:decompRules} can be applied directly to this LTS. The top pattern applies twice, once per state, and the bottom pattern applies three times, once per transition.

\begin{figure}
    \centering
\begin{tikzpicture}[->,shorten >=1pt,auto,node distance=2cm,on grid,semithick, inner sep=0.3mm, bend angle=20, transform shape
]
\node [label={[xshift=0.5cm,yshift=-4pt]$\underset{\{a_1,a_2\},\{b_1\}}{\xrightarrow{decomp_s}}$}] (decomp) { };
\node[initial,initial where=above,initial text=,state,inner sep=1mm, minimum size=0pt] (r) [above left= of decomp]                    {$r$};
\node[state,inner sep=1mm, minimum size=0pt]         (s) [below left=of decomp] {$s$};
\path [every node/.style={font=\footnotesize}, bend angle=37]
(r) edge              node [anchor=east] {$a_1$} (s)
(s) edge [bend left]  node {$a_2$} (r)
    edge [bend right] node [anchor=west] {$b_1$} (r)
;
\node[state, elliptic state, inner sep=0.3mm, minimum size=0pt]         (tr1) [right=of decomp] {$t_{a_2,r_\uparrow}$};
\node[state, elliptic state, inner sep=0.3mm, minimum size=0pt]         (ts1) [right= 1.4cm of tr1] {$t_{a_1,s_\uparrow}$};
\node[initial,initial where=above,initial text=,state,inner sep=0.3mm, minimum size=0pt] (r1)                   [above =of ts1] {$r_\uparrow$};
\node[state,inner sep=0.3mm, minimum size=0pt]         (r2) [right=of r1] {$r_\downarrow$};
\node[state,inner sep=0.3mm, minimum size=0pt]         (s1) [below =of ts1] {$s_\uparrow$};
\node[state,inner sep=0.3mm, minimum size=0pt]         (s2) [right=of s1] {$s_\downarrow$};

\path [every node/.style={font=\footnotesize}, bend angle=17]
(r1) edge              node {$a_1$} (ts1)
     edge  [bend left] node {$c_{r_\uparrow,r_\downarrow}$} (r2)
(r2) edge  [bend left] node {$\overline{c_{r_\downarrow,r_\uparrow}}$} (r1)
(ts1) edge node {$t_{s_\uparrow}$} (s1)
(s1) edge [bend left=30]  node {$a_2$} (tr1)
     edge [bend left] node {$c_{s_\uparrow,s_\downarrow}$} (s2)
(tr1) edge [bend left=30] node {$t_{r_\uparrow}$} (r1)
(s2) edge node [anchor=west] {$\overline{t_{r_\downarrow}}$} (r2)
     edge [bend left] node {$\overline{c_{s_\downarrow,s_\uparrow}}$} (s1)
;

\node[] (par) [right= 18pt of $(r2)!0.5!(s2)$] {\Huge{$\mid$}};
\node[initial,initial where=above,initial text=,state,inner sep=0.3mm, minimum size=0pt] (r1_2)                   [right =of r2] {$r_\uparrow$};
\node[state,inner sep=0.3mm, minimum size=0pt]         (r2_2) [right=of r1_2] {$r_\downarrow$};
\node[state,inner sep=0.3mm, minimum size=0pt]         (s1_2) [right =of s2] {$s_\uparrow$};
\node[state,inner sep=0.3mm, minimum size=0pt]         (s2_2) [right=of s1_2] {$s_\downarrow$};
\node[state, elliptic state,inner sep=0.3mm, minimum size=0pt]         (tr2) at ($(r2_2)!0.5!(s2_2)$) {$t_{b_1,r_\downarrow}$};

\path [every node/.style={font=\footnotesize}, bend angle=17]
(r1_2) edge              node {$\overline{t_{s_\uparrow}}$} (s1_2)
     edge  [bend left] node {$\overline{c_{r_\uparrow,r_\downarrow}}$} (r2_2)
(r2_2) edge  [bend left] node {$c_{r_\downarrow,r_\uparrow}$} (r1_2)
(tr2) edge [anchor=west] node {$t_{r_\downarrow}$} (r2_2)
(s1_2) edge [bend left]  node {$\overline{t_{r_\uparrow}}$} (r1_2)
     edge [bend left] node {$\overline{c_{s_\uparrow,s_\downarrow}}$} (s2_2)
(s2_2) edge node [anchor=west] {$b_1$} (tr2)
     edge [bend left] node {$c_{s_\downarrow,s_\uparrow}$} (s1_2)
;

\end{tikzpicture}
    \caption{Example of synchronous decomposition operation of \cref{dfn:decomposition}.}
    \label{fig:decompExample}
\end{figure}
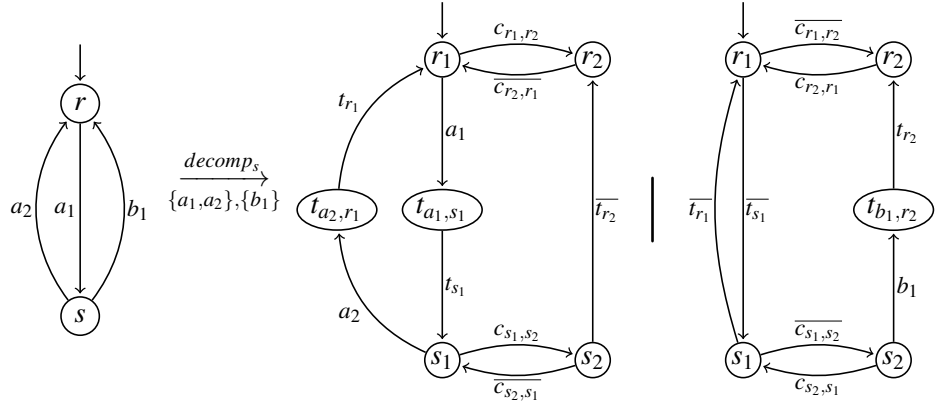

Next, we compute the synchronous product and form one LTS shown at the right of \cref{fig:bbexample}. The nodes are divided into two equivalence classes, top and bottom. The states in the top class are branching bisimilar to state $r$ whereas the states in the bottom one are branching bisimilar to state $s$.

The following proves the branching bisimilarity and thus proves that there is a way of decomposing an LTS such that branching bisimilarity is maintained.

\begin{thm}
\label{lem:bbproof}
Given an LTS $M = (S, \Sigma, \rightarrow, s_0)$ and two alphabets $\Sigma_1, \Sigma_2$ such that $\Sigma = \Sigma_1 \cup \Sigma_2$ and $\Sigma_1 \cap \Sigma_2 = \emptyset$, and given an LTS $M_x = M_1 \times M_2$ where $(M_1, M_2) = decomp_s(M)$ by \cref{dfn:decomposition}, then $M \leftrightarroweq_b M_x$.
\end{thm}

\begin{proof}

Let $M_1 = (S_C \cup S_{T_1}, \Sigma_1 \cup \Sigma_{S_1}, \rightarrow_1, q_1)$ and $M_2 = (S_C \cup S_{T_2}, \Sigma_2 \cup \Sigma_{S_2}, \rightarrow_2, q_2)$.

Define a relation $\mathcal{R} \subseteq S \times ((S_C \cup S_{T_1}) \times (S_C \cup S_{T_2}))$ with $\mathcal{R} = \{\langle s, (\qs{s}{n},\qs{s}{n}) \rangle$, $\langle r', (t_{a,\qs{r'}{n}},\qs{r}{n}) \rangle$, $\langle r', (\qs{r}{n},t_{a,\qs{r'}{n}}) \rangle$ $\mid s,r,r' \in S, n \in \{\uparrow,\downarrow\}$, $a \in \Sigma$, 
$r \xrightarrow{a} r' \}$.
We prove that $\mathcal{R}$ is a branching bisimulation relation through the following cases:
\begin{enumerate}
\item Consider a pair $\langle s, s_x \rangle = \langle s, (\qs{s}{n},\qs{s}{n}) \rangle$ where $n \in \{\uparrow,\downarrow\}$.
\begin{enumerate}
    \item Assume $s \xrightarrow{a} s'$. Then we have two cases: 
    \begin{enumerate}
        \item $a \in \Sigma_1$. Then, by~(\ref{eq3}), $s_x \xrightarrow{a}_x (t_{a,\qs{s'}{n}},\qs{s}{n})$. We see that $\langle s', (t_{a,\qs{s'}{n}},\qs{s}{n}) \rangle \in \mathcal{R}$.
        \item $a \in \Sigma_2$. Then, by~(\ref{eq3}), $s_x \xrightarrow{a}_x (\qs{s}{n},t_{a,\qs{s'}{n}})$. We see that $\langle s', (\qs{s}{n},t_{a,\qs{s'}{n}}) \rangle \in \mathcal{R}$.
        \end{enumerate}
        \item Assume $s_x \xrightarrow{a}_x s_x'$. Then we have the following three cases: 
        \begin{enumerate}
        \item $a \in \Sigma_1 \land \overline{a} \not\in \Sigma_2$, then this transition is only possible, by definition, through the transition $s \xrightarrow{a} s'$ for some $s'$ such that $s_x' \overset{\mathrm{(\ref{eq3})}}{=} (t_{a,s'},s_\uparrow)$. We see that $\langle s',s_x' \rangle \in \mathcal{R}$.
        \item $a \in \Sigma_2 \land \overline{a} \not\in \Sigma_1$. This is a symmetric case where $s \xrightarrow{a} s'$ and $s_x' \overset{\mathrm{(\ref{eq3})}}{=} (s_\downarrow,t_{a,s'})$.  We see that $\langle s',s_x' \rangle \in \mathcal{R}$.
        \item $a \in \Sigma_1 \land \overline{a} \in \Sigma_2$, then the only transition possible is the $\tau$ transition of~(\ref{eq5}). Then $s_x' = (\qs{s}{m},\qs{s}{m})$ where $m \in \{\uparrow,\downarrow\}$ and $m \neq n$. We see that $\langle s, s_x' \rangle \in \mathcal{R}$.
    \end{enumerate}
\end{enumerate}

\item Consider a pair $\langle r, r_x \rangle = \langle s', (t_{a,\qs{s'}{n}},\qs{s}{n}) \rangle$ where $n \in \{\uparrow,\downarrow\}$, $a \in \Sigma$ and $s \xrightarrow{a} s'$.
\begin{enumerate}
\item Assume $r \xrightarrow{a} r'$. Then we show that $r_x \xrightarrow{\tau}_x r_x'$ and $r_x' \xrightarrow{a}_x r_x''$ and $\langle r, r_x' \rangle \in \mathcal{R}$ and $\langle r', r_x'' \rangle \in \mathcal{R}$. We do this for $a \in \Sigma_1$. The case for $a \in \Sigma_2$ is symmetric.
\begin{enumerate}
\item $r_x' \overset{\mathrm{(\ref{eq3})}}{=} (s'_\downarrow,s'_\downarrow) $. We see that $\langle r, r_x' \rangle \in \mathcal{R}$.
\item Since $r = s'$ and $r \xrightarrow{a} r'$, then by~(\ref{eq3}), there exists a state $r_x''$ such that $r_x' \xrightarrow{a}_x r_x''$, and $r_x'' = (t_{a,s''_\downarrow},s'_\downarrow)$, where $s'' = r'$. We see that $\langle r', r_x'' \rangle \in \mathcal{R}$.
\end{enumerate}
\item Assume $r_x \xrightarrow{a}_x r_x'$. By the definition of $\rightarrow_x$, it is only possible that $a$ is a $\tau$ action and that $n=\uparrow$. Thus, $r_x' \overset{\mathrm{(\ref{eq3})}}{=} (s'_\downarrow,s'_\downarrow)$. We see that $\langle r', r_x'' \rangle \in \mathcal{R}$.
\end{enumerate}
\item Consider a pair $\langle r, r_x \rangle = \langle s', (\qs{s}{n},t_{a,\qs{s'}{n}}) \rangle$ where $n \in \{\uparrow,\downarrow\}$, $a \in \Sigma$ and $s \xrightarrow{a} s'$. This case is symmetric to Case 2.
\end{enumerate}

\end{proof}

\begin{corollary}
\label{thm:bb}
It follows from \cref{lem:bbproof} that there is a universal way of decomposing an LTS $M$ using a general synchronous decomposition operator (\cref{dfn:decomposition_general}) such that $M$ is branching-bisimilar to the synchronous product of its decomposition.
\end{corollary}


\begin{figure}
    \centering
\begin{tikzpicture}[->,shorten >=1pt,auto,node distance=2cm,on grid,semithick, inner sep=1mm, minimum size=0pt,bend angle=20, 
    line/.style={
      draw,thick,
      -latex',
      shorten >=2pt
    }]
\tikzstyle{rectangular} = [draw,rectangle,rounded corners]
\node [scale=1.4] (decomp) {$\leftrightarroweq_{b}$};
\node[initial,initial where=above,initial text=,state, minimum size=0pt] (r) [above left= 1.25cm and 1.7cm of decomp]                    {$r$};
\node[state,minimum size=0pt]         (s) [below left= 1.25cm and 1.7cm of decomp] {$s$};
\path [every node/.style={font=\footnotesize}, bend angle=37]
(r) edge              node [anchor=east] {$a_1$} (s)
(s) edge [bend left]  node {$a_2$} (r)
    edge [bend right] node [anchor=west] {$b_1$} (r)
;
\node[rectangular] (tr1) [right=of decomp] {$(t_{a_2,r_\uparrow},s_\uparrow)$};
\node[rectangular] (ts1) [right= 2.7cm of tr1] {$(t_{a_1,s_\uparrow},r_\uparrow)$};
\node[initial,initial where=above,initial text=,draw,rectangle,rounded corners, minimum size=0pt] (r1)                   [above =of ts1] {$(r_\uparrow,r_\uparrow)$};
\node[draw,rectangle,rounded corners, minimum size=0pt]         (r2) [right= 2.7cm of r1] {$(r_\downarrow,r_\downarrow)$};
\node[draw,rectangle,rounded corners, minimum size=0pt]         (s1) [below =of ts1] {$(s_\uparrow,s_\uparrow)$};
\node[draw,rectangle,rounded corners, minimum size=0pt]         (s2) [right= 2.7cm of s1] {$(s_\downarrow,s_\downarrow)$};
\node[draw,rectangle,rounded corners]         (tr2) at ($(r2)!0.5!(s2)$) {$(s_\downarrow,t_{b_1,r_\downarrow})$};
\draw [red,thick,dotted] 
   ($ (tr1.south west) + (-0.2,-0.3) $) -- 
   ($ (tr1.north west)+(-0.2,0.2) $) |- 
   ($ (r1.north west)+(-0.2,0.5) $) -- 
   ($ (r2.north east)+(0.4,0.5) $) |- 
   ($ (tr2.south east)+(0.18,-0.3) $) -- 
   ($ (tr2.south west)+(-0.2,-0.3) $) -- 
   ($ (tr2.north west)+(-0.2,1) $) -- 
   ($ (tr1.north east)+(0.2,1) $) -- 
   ($ (tr1.south east)+(0.2,-0.3) $) -- 
   cycle;
\draw [densely dashed] 
   ($ (ts1.south west)+(-0.2,-0.2) $) -- 
   ($ (ts1.north west)+(-0.2,0.25) $) -- 
   ($ (ts1.north east)+(0.2,0.25) $) -- 
   ($ (ts1.south east)+(0.2,-0.2) $) |- 
   ($ (s2.north east)+(0.2,0.4) $) -- 
   ($ (s2.south east)+(0.2,-0.4) $) -- 
   ($ (s1.south west)+(-0.25,-0.4) $) -| 
   cycle;

\path [every node/.style={font=\footnotesize}, bend angle=14]
(r1) edge              node {$\boldsymbol{a_1}$} (ts1)
     edge  [bend left] node {$\tau$} (r2)
(r2) edge  [bend left] node {$\tau$} (r1)
(ts1) edge node {$\tau$} (s1)
(s1) edge [bend left=40]  node {$\boldsymbol{a_2}$} (tr1)
     edge [bend left] node {$\tau$} (s2)
(tr1) edge [bend left=40] node {$\tau$} (r1)
(s2) edge node [swap] {$\boldsymbol{b_1}$} (tr2)
     edge [bend left] node {$\tau$} (s1)
(tr2) edge [swap] node {$\tau$} (r2)
;
\end{tikzpicture}
    \caption{Showing branching bisimulation on the example of \cref{fig:decompExample}.}
    \label{fig:bbexample}
\end{figure}
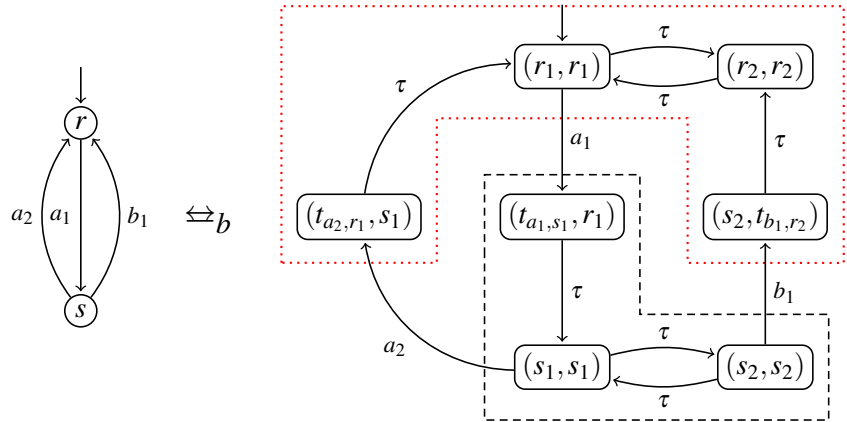

\subsection{Decomposing into Asynchronous Subsystems}
We consider asynchronous communication between subsystems simply because many practical systems use asynchronous communication and the aim is to extend our proof towards that.
So, we define a new decomposition operation ($decomp_a$) such that the communication between subsystems is asynchronous. We assign each subsystem a queue that stores received messages until they are consumed. An action of sending such a message does synchronise, however, with the queue of the opposite side receiving it.
The operation $decomp_a$ is summarised in \cref{fig:decompRules_async}.

\begin{figure}
    \centering
\begin{tikzpicture}[->,shorten >=1pt,auto,on grid,semithick, inner sep=0.3mm,node distance=2cm, minimum size=0pt,bend angle=20,  transform shape]
\node [state,inner sep=1mm, minimum size=0pt] (r) {$r$};
\node [state,inner sep=0.3mm, minimum size=0pt,draw=none] (decomp) [right= of r] {$\xrightarrow{decomp_a}$};
\node[state,inner sep=1mm, minimum size=0pt]         (s') [below right= 2cm and 0.75cm of r] {$s$};
\node[state,inner sep=1mm, minimum size=0pt]         (s) [below left= 2cm and 0.75cm of r] {$r$};
\node [state,inner sep=0.3mm, minimum size=0pt,draw=none] (decomp2) [right= of s'] {$\xrightarrow{decomp_a}_{a \in \Sigma_1}$};

\path [every node/.style={font=\footnotesize}, bend angle=37]
(s) edge node {$a$} (s')
;

\node[state,rectangle,rounded corners, minimum size=0pt] (r1) [right =of decomp] {$\Lbag r_\uparrow, \epsilon \Rbag$};
\node[state,rectangle,rounded corners, minimum size=0pt] (r2) [right= 2.8cm of r1] {$\Lbag r_\downarrow, \epsilon \Rbag$};
\node[state, elliptic state, rectangle,rounded corners, minimum size=0pt] (r2tr1) [above right= 1cm and 1.4cm of r1] {$\Lbag r_\downarrow, t_{r_\uparrow} \Rbag$};

\node[state, elliptic state, rectangle,rounded corners, minimum size=0pt] (s1) [right =of decomp2] {$\Lbag r_\uparrow, Q_1 \Rbag$};
\node[state, elliptic state, rectangle,rounded corners, minimum size=0pt] (tr1) [above right= 1cm and 1cm of s1] {$\Lbag t_{a,s_\uparrow}, Q_1 \Rbag$};
\node[state, elliptic state, rectangle,rounded corners, minimum size=0pt] (s2) [below right= 1cm and 1cm of tr1] {$\Lbag s_\uparrow, Q_1 \Rbag$};

\path [every node/.style={font=\footnotesize}, bend angle=17]
(r1) edge  [bend right] node {$c_{r_{1,2}}$} (r2)
(r2) edge  [bend right=32] node {$\overline{c_{r_{2,1}}}$} (r2tr1)
(r2tr1) edge [bend right=25] node {$\tau$} (r1)
(s1) edge [bend left=25] node {$a$} (tr1)
(tr1) edge [bend left=25] node {$t_{s_\uparrow}$} (s2)
;

\node[] (par) [right= 1cm of r2] {\Huge{$\mid$}};
\node[] (par2) [right= 1.1cm of s2] {\Huge{$\mid$}};

\node[state,rectangle,rounded corners, minimum size=0pt] (r1_2) [right =of r2] {$\Lbag r_\uparrow, \epsilon \Rbag$};
\node[state,rectangle,rounded corners, minimum size=0pt] (r2_2) [right= 2.8cm of r1_2] {$\Lbag r_\downarrow, \epsilon \Rbag$};
\node[state, elliptic state, rectangle,rounded corners, minimum size=0pt] (r1tr2) [above right= 1cm and 1.4cm of r1_2] {$\Lbag r_\uparrow, t_{r_\downarrow} \Rbag$};

\node[state, elliptic state, rectangle,rounded corners, minimum size=0pt] (s1_2) [right =2.3cm of s2] {$\Lbag r_\uparrow, Q_2 \Rbag$};
\node[state, elliptic state, rectangle,rounded corners, minimum size=0pt] (tr1_2) [above = 1cm of s1_2] {$\Lbag r, t_{s_\uparrow} \cdot Q_2 \Rbag$};
\node[state, elliptic state, rectangle,rounded corners, minimum size=0pt] (tr1_3) [right= 2cm of tr1_2] {$\Lbag r, Q_2 \cdot t_{s_\uparrow} \Rbag$};
\node[state, elliptic state, rectangle,rounded corners, minimum size=0pt] (s2_2) [below = 1cm of tr1_3] {$\Lbag s_\uparrow, Q_2 \Rbag$};

\path [every node/.style={font=\footnotesize}, bend angle=17]
(r1_2) edge  [swap, bend left=25] node {$\overline{c_{r_{1,2}}}$} (r1tr2)
(r1tr2) edge [swap, bend left=25] node {$\tau$} (r2_2)
(r2_2) edge  [swap, bend left] node {$c_{r_{2,1}}$} (r1_2)
(s1_2) edge [bend left=25] node {$\overline{t_{s_\uparrow}}$} (tr1_2)
(tr1_3) edge [bend left=25] node {$\tau$} (s2_2)
;

\end{tikzpicture}
    \caption{The two patterns that delineate the $decomp_a$ operator (\cref{dfn:decompositionasync2}).}
    \label{fig:decompRules_async}
\end{figure}
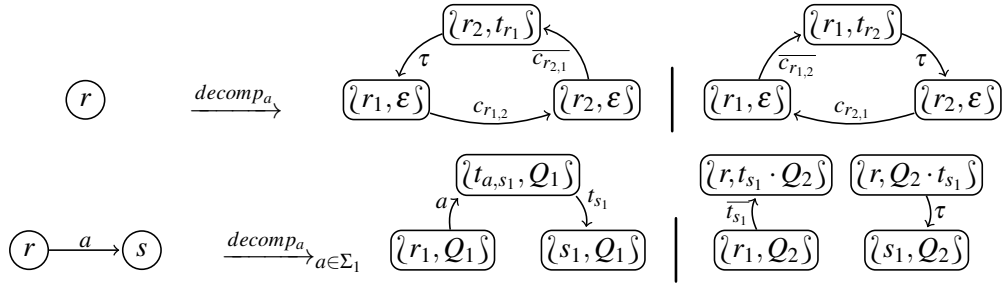





\begin{dfn} [LTS with Queue]
\label{dfn:lts-with-queue}
A queue is an ordered-list of actions. An LTS with a queue is a transition system of the shape $(S \times Q,\Sigma, \rightarrow, s_0)$. A state in~$S\times Q$ holds the contents of the queue~$Q$ and is written as~$\Lbag s, Q \Rbag$. 
\end{dfn}

Elements in a queue are concatenated using the~$\cdot$ operator. Appending an element~$m$ to the back of a queue~$Q$ produces the queue~$m \cdot Q$, while $Q \cdot m$ represents the queue with $m$ in the front. The symbol $\epsilon$ represents the empty queue. 

\begin{dfn} [Decomposing into asynchronous subsystems]
Given an LTS $M = (S, \Sigma, \rightarrow, r_0)$ and two alphabets $\Sigma_1, \Sigma_2$ such that $\Sigma = \Sigma_1 \cup \Sigma_2$ and $\Sigma_1 \cap \Sigma_2 = \emptyset$, then we can decompose $M$ over $\Sigma_1$ and $\Sigma_2$ by applying the following operation:



$decomp_a (M, \Sigma_1, \Sigma_2) = (M_1, M_2)$ where, for $i,j \in \{\uparrow,\downarrow\}$ and $i \neq j$, $M_i$ is an LTS with a queue (\cref{dfn:lts-with-queue}) defined as follows:
\begin{enumerate}
\item $M_i = ((S_C \cup S_{T_i}, Q_i), \Sigma_i \cup \Sigma_{S_i}, \rightarrow_i, r_1)$
\item For every state in $S$, we introduce a pair of states $s_\uparrow, s_\downarrow \in S_C$, a pair of $c$-actions, and a pair of $t$-actions: 
\begin{equation}
\begin{aligned}
&S_{C_i} = \{\qs{s}{i} \mid s \in S\} 
\hspace{1cm} S_C = S_{C_1} \cup S_{C_2}\\
&\Sigma_{C_{i,j}} = \{c_{s_{i,j}} \mid s \in S\} 
\hspace{0.4cm} \Sigma_{T_i} = \{t_{\qs{s}{i}} \mid s \in S \} 
\end{aligned}
\end{equation}
\item Sets of $t$-states are defined as follows:
    \begin{equation}
      \begin{aligned}
    &S_{T_i} = \{t_{a,\qs{s}{i}} \mid a \in \Sigma_i, \qs{s}{i} \in S_{C_i}\} \\
      \end{aligned}
    \end{equation}

\item Sets of synchronous actions are defined as follows: 
\begin{equation}
  \begin{aligned}
&\Sigma_{S_i} = \Sigma_{T_i} \cup \overline{\Sigma_{T_j}} \cup \Sigma_{C_{i,j}} \cup \overline{\Sigma_{C_{j,i}}} \\
  \end{aligned}
\end{equation}


\item The transition relation ${\rightarrow_i} \subseteq (S_C \cup S_{T_i}) \times Q_i \times (\Sigma_i \cup \Sigma_{S_i}) \times (S_C \cup S_{T_i}) \times Q_i$ is the minimal relation satisfying the following:

\begin{enumerate}
\item For all $s \in S$:
\begin{equation}
	\Lbag \qs{s}{i}, Q_i \Rbag \xrightarrow{c_{s_{i,j}}}_i \Lbag \qs{s}{j}, Q_i \Rbag \\
    \Lbag \qs{s}{i}, Q_i \Rbag \xrightarrow{\overline{c_{s_{i,j}}}}_j \Lbag \qs{s}{i}, t_{\qs{s}{j}} \cdot Q_i \Rbag \\
    \label{eq:a-transition-1}
\end{equation}

\item For all $s,s' \in S$, and all $a \in \Sigma_i$, if $s \xrightarrow{a} s'$, then: 
\begin{equation}
  \begin{aligned}
     \Lbag \qs{s}{i}, Q_i \Rbag \xrightarrow{a}_i \Lbag t_{a, \qs{s'}{i}}, Q_i \Rbag \xrightarrow{t_{\qs{s'}{i}}}_i \Lbag \qs{s'}{i}, Q_i \Rbag\\
     \Lbag \qs{s}{i}, Q_i \Rbag \xrightarrow{\overline{t_{\qs{s'}{i}}}}_j \Lbag \qs{s}{i},t_{\qs{s'}{i}} \cdot Q_i \Rbag \\
     \label{eq:a-transition-2}
  \end{aligned}
\end{equation}

\item Consuming an element from the front of a queue is an internal transition of the form:
\begin{equation}
\begin{aligned}
    \Lbag s, Q \cdot t_{s'} \Rbag \xrightarrow{\tau} \Lbag s', Q \Rbag \\
    \label{eq:a-rcv}
\end{aligned}
\end{equation}

\end{enumerate}
\end{enumerate}

\label{dfn:decompositionasync2}
\end{dfn}

We see in~(\ref{eq:a-transition-1}) that the two automata synchronise on action $c_{s_{i,j}}$. The effect is a message sent from $M_i$ and received in the queue of $M_j$. The same occurs in~(\ref{eq:a-transition-2}). Moreover, this makes sending messages only possible when both machines are in sync, i.e., on the same state $\qs{s}{i}$. This model is inspired by asynchronous communication in practice and we found it necessary that a sent message is received before any other actions are performed by either side.

The question may arise ``why do we still see synchronisation? Is this still considered asynchronous communication?''; and the answer is that it is asynchronous in the sense that the sending party does not know whether and when the receiving party is willing to receive the communication. This is different from synchronous communication where the sender knows that the receiver is willing to participate. 

For this construction, a queue of size 1 is enough as it never contains more than one message. If the size of the queue was more than 1, then the transition in~(\ref{eq:a-transition-2}) can apply recursively until the queue is full. However, any transition beyond the first one does not translate into a transition in the product because action $c_{s_{i,j}}$ needs to synchronise with its co-action. In other words, in the product automaton, the queue cannot contain more than one message. Therefore, the upperbound of size 1 of the queue is not a requirement, but rather follows from the construction.


\subsection{Proof that the Asynchronous Decomposition Operation Maintains Branching Bisimulation}
\label{sec:decomp-async-proof}

In this subsection, similar to \cref{sec:decomp-sync-proof}, we prove that the asynchronous decomposition operation ($decomp_a$) also maintains branching bisimilarity.  \cref{fig:decompExample2} shows the result of applying $decomp_a$ to the same example behaviour as in \cref{fig:decompExample}. In \cref{fig:async-bb}, we compute the synchronous product of the decomposition of \cref{fig:decompExample2} and then divide the nodes of the product into two equivalence classes, top and bottom. The states in the top class are branching bisimilar to state $r$ whereas the states in the bottom on are branching bisimilar to state $s$.

\begin{figure}
    \centering
\begin{tikzpicture}[->,shorten >=1pt,auto,node distance=2cm,on grid,semithick, inner sep=0.3mm, bend angle=20, transform shape
]
\node [label={[xshift=0.5cm,yshift=-4pt]$\underset{\{a_1,a_2\},\{b_1\}}{\xrightarrow{decomp_a}}$}] (decomp) { };
\node[initial,initial where=above,initial text=,state,inner sep=1mm, minimum size=0pt] (r) [above left= of decomp]                    {$r$};
\node[state,inner sep=1mm, minimum size=0pt]         (s) [below left=of decomp] {$s$};
\path [every node/.style={font=\footnotesize}, bend angle=37]
(r) edge              node [anchor=east] {$a_1$} (s)
(s) edge [bend left]  node {$a_2$} (r)
    edge [bend right] node [anchor=west] {$b_1$} (r)
;
\node[state, rectangle, rounded corners, minimum size=0pt] (tr1) [right= 2.3cm of decomp] {$\Lbag t_{a_2,r_\uparrow}, \epsilon \Rbag$};
\node[state, rectangle,rounded corners, minimum size=0pt] (ts1) [right= 1.8cm of tr1] {$\Lbag t_{a_1,s_\uparrow}, \epsilon \Rbag$};
\node[state, rectangle,rounded corners, minimum size=0pt] (str2) [right= of ts1] {$\Lbag s_\downarrow, t_{r_\downarrow} \Rbag$};
\node[initial,initial where=left,initial text=,state,rectangle,rounded corners, minimum size=0pt] (r1) [above right = 1.5cm and 0.9cm of tr1] {$\Lbag r_\uparrow, \epsilon \Rbag$};
\node[state,rectangle,rounded corners, minimum size=0pt] (s1) [below right = 1.5cm and 0.9cm of tr1] {$\Lbag s_\uparrow, \epsilon \Rbag$};
\node[state,rectangle,rounded corners, minimum size=0pt] (s2) [below= 1.5cm of str2] {$\Lbag s_\downarrow, \epsilon \Rbag$};
\node[state,rectangle,rounded corners, minimum size=0pt] (r2) [above= 1.5cm of str2] {$\Lbag r_\downarrow, \epsilon \Rbag$};
\node[state, elliptic state, rectangle,rounded corners, minimum size=0pt] (r2tr1) [above right= 1cm and 1.4cm of r1] {$\Lbag r_\downarrow, t_{r_\uparrow} \Rbag$};
\node[state, elliptic state, rectangle,rounded corners, minimum size=0pt] (s2ts1) [below right= 1cm and 1.4cm of s1] {$\Lbag s_\downarrow, t_{s_\uparrow} \Rbag$};

\path [every node/.style={font=\footnotesize}, bend angle=17]
(r1) edge [bend left=17]  node {$a_1$} (ts1)
     edge  [bend right] node {$c_{r_{1,2}}$} (r2)
(r2) edge  [bend right=32] node {$\overline{c_{r_{2,1}}}$} (r2tr1)
(r2tr1) edge [bend right=25] node {$\tau$} (r1)
(ts1) edge [bend left=17] node {$t_{s_\uparrow}$} (s1)
(s1) edge [bend left=17]  node {$a_2$} (tr1)
     edge [swap, bend left] node {$c_{s_{1,2}}$} (s2)
(tr1) edge [bend left=17] node {$t_{r_\uparrow}$} (r1)
(s2) edge node [swap] {$\overline{t_{r_\downarrow}}$} (str2)
     edge [swap, bend left=25] node {$\overline{c_{s_{2,1}}}$} (s2ts1)
(s2ts1) edge [swap, bend left=25] node {$\tau$} (s1)
(str2) edge [swap] node {$\tau$} (r2)
;

\node[] (par) [right= 35pt of $(r2)!0.5!(s2)$] {\Huge{$\mid$}};
\node[initial,initial where=left,initial text=,state, rectangle,rounded corners, minimum size=0pt] (r1_2) [right =3.7cm of r2] {$\Lbag r_\uparrow, \epsilon \Rbag $};
\node[state,rectangle,rounded corners, minimum size=0pt] (r2_2) [right = 2.7 cm of r1_2] {$\Lbag r_\downarrow, \epsilon \Rbag$};
\node[state, rectangle,rounded corners, minimum size=0pt] (s1_2) [right = 3.7cm of s2] {$\Lbag s_\uparrow, \epsilon \Rbag$};
\node[state, rectangle,rounded corners, minimum size=0pt] (s2_2) [right = 2.7 cm of s1_2] {$\Lbag s_\downarrow, \epsilon \Rbag$};
\node[state, elliptic state, rectangle,rounded corners, minimum size=0pt] (tr2) at ($(r2_2)!0.5!(s2_2)$) {$\Lbag t_{b_1,r_\downarrow}, \epsilon \Rbag$};
\node[state, elliptic state, rectangle,rounded corners, minimum size=0pt] (s1ts2) [below right= 1cm and 1.4cm of s1_2] {$\Lbag s_\uparrow, t_{s_\downarrow} \Rbag$};
\node[state, elliptic state, rectangle,rounded corners, minimum size=0pt] (r1tr2) [above right= 1cm and 1.4cm of r1_2] {$\Lbag r_\uparrow, t_{r_\downarrow} \Rbag$};
\node[state, elliptic state, rectangle,rounded corners, minimum size=0pt] (r1ts1) [right=  4.6cm of str2] {$\Lbag r_\uparrow, t_{s_\uparrow} \Rbag$};
\node[state, elliptic state, rectangle,rounded corners, minimum size=0pt] (s1tr1) [right= 2.8cm of str2] {$\Lbag s_\uparrow, t_{r_\uparrow} \Rbag$};

\path [every node/.style={font=\footnotesize}, bend angle=17]
(r1_2) edge [bend left] node {$\overline{t_{s_\uparrow}}$} (r1ts1)
     edge  [swap, bend left=25] node {$\overline{c_{r_{1,2}}}$} (r1tr2)
(r1ts1) edge [bend left] node {$\tau$} (s1_2)
(r1tr2) edge [swap, bend left=25] node {$\tau$} (r2_2)
(r2_2) edge  [swap, bend left] node {$c_{r_{2,1}}$} (r1_2)
(tr2) edge [swap] node {$t_{r_\downarrow}$} (r2_2)
(s1_2) edge [bend left]  node {$\overline{t_{r_\uparrow}}$} (s1tr1)
     edge [bend right=25] node {$\overline{c_{s_{1,2}}}$} (s1ts2)
(s1tr1) edge [bend left] node {$\tau$} (r1_2)
(s1ts2) edge [bend right=25] node {$\tau$} (s2_2)
(s2_2) edge node [swap] {$b_1$} (tr2)
     edge [bend right] node {$c_{s_{2,1}}$} (s1_2)
;

\end{tikzpicture}
    \caption{Example of asynchronous decomposition operation (\cref{dfn:decompositionasync2}).}
    \label{fig:decompExample2}
\end{figure}
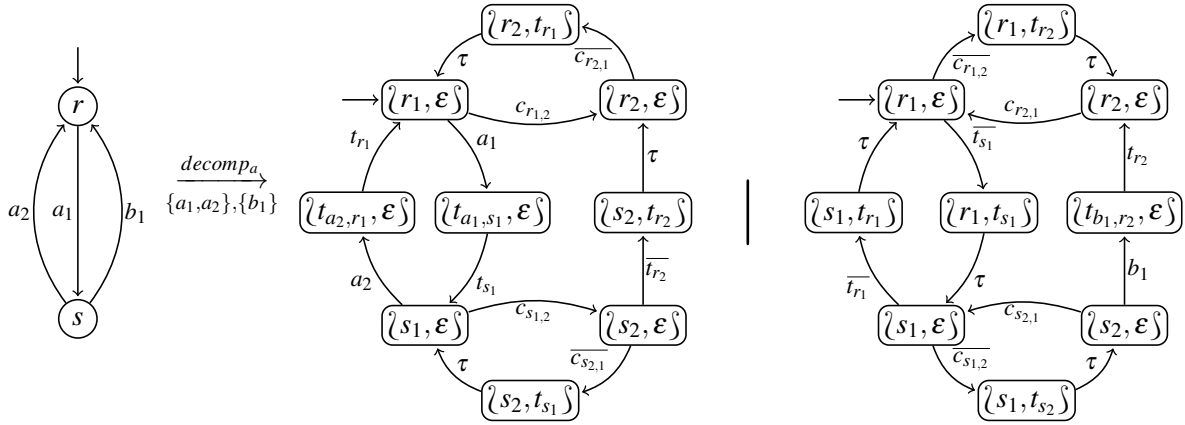

Next, we prove that any LTS decomposed using \cref{dfn:decompositionasync2} maintains branching bisimulation with its decomposition, thus by proving that there is at least one universal method of decomposing LTSs into asynchronous ones while maintaining branching bisimulation. 

\begin{thm}
\label{lem:bbproof-async}
Given an LTS $M = (S, \Sigma, \rightarrow, s_0)$ and two alphabets $\Sigma_1, \Sigma_2$ such that $\Sigma = \Sigma_1 \cup \Sigma_2$ and $\Sigma_1 \cap \Sigma_2 = \emptyset$, and given an LTS $M_x = M_1 \times M_2$ where $(M_1, M_2) = decomp_a(M)$ by \cref{dfn:decompositionasync2}, then $M \leftrightarroweq_b M_x$.
\end{thm}

\begin{proof}

Let $M_1 = ((S_C \cup S_{T_1}, Q_1), \Sigma_1 \cup \Sigma_{S_1}, \rightarrow_1, \qs{r}{1})$ and $M_2 = ((S_C \cup S_{T_2}, Q_2), \Sigma_2 \cup \Sigma_{S_2}, \rightarrow_2, \qs{r}{1})$.

Define a relation $\mathcal{R} \subseteq S \times ((S_C \cup S_{T_1}, Q_1) \times (S_C \cup S_{T_2}, Q_2))$ with $\mathcal{R} = $
\[
\begin{array}{l@{}}
\{\langle s, (\Lbag s_i, \epsilon \Rbag,\Lbag s_i, \epsilon \Rbag) \rangle, \\

\langle s, (\Lbag s_i, t_{s_j}\Rbag, \Lbag s_j, \epsilon \Rbag) \rangle,
\langle s, (\Lbag s_j, \epsilon \Rbag, \Lbag s_i, t_{s_j}\Rbag) \rangle,\\

\langle s, (t_{a,s_i}, \Lbag r_i, \epsilon \Rbag ) \rangle,
\langle s, (\Lbag r_i, \epsilon \Rbag, t_{a,s_i}) \rangle,\\
\langle s, (\Lbag s_i, \epsilon \Rbag, \Lbag r_i, t_{s_i} \Rbag) \rangle,
\langle s, (\Lbag r_i, t_{s_i} \Rbag, \Lbag s_i, \epsilon \Rbag) \rangle,\\
\langle u, (t_{b,s'_i}, \Lbag r_i, t_{s_i} \Rbag) \rangle,
\langle u, (\Lbag r_i, t_{s_i} \Rbag, t_{b,s'_i}) \rangle \mid\\
r,s,u \in S \text{ and } i,j \in \{\uparrow,\downarrow\}\text{ where }i \neq j\text{ and }a,b \in \Sigma_i\text{ and }
r \xrightarrow{a} s \xrightarrow{b} u
 \}.\\
\end{array}
\]
We prove that $\mathcal{R}$ is a branching bisimulation relation through the following cases:
\begin{enumerate}
    \item \label{prf:case-1} Consider a pair $\langle s, s_x \rangle = \langle s, (\Lbag s_i, \epsilon \Rbag,\Lbag s_i, \epsilon \Rbag) \rangle$ where $i \in \{1,2\}$.
    \begin{enumerate}
        \item Assume $s \xrightarrow{a} s'$. Then if $a \in \Sigma_1$, then $s_x \xrightarrow{a}_1 s'_x$ where $s'_x \overset{(\ref{eq:a-transition-2})}{=} (t_{a,s'_i},\Lbag s_i, \epsilon \Rbag)$ with $i=1$. Else if $a \in \Sigma_2$ then $s_x \xrightarrow{a}_2 s''_x$ where $s''_x \overset{(\ref{eq:a-transition-2})}{=} (\Lbag s_i, \epsilon \Rbag,t_{a,s'_i})$ with $i=2$. We see that both pairs $\langle s, s'_x \rangle$ and $\langle s, s''_x \rangle$ are in $\mathcal{R}$.  
        \item Assume $s_x \xrightarrow{a}_x s'_x$. Then we have the following three cases: 
        \begin{enumerate}
        \item $a \in \Sigma_1 \land \overline{a} \not\in \Sigma_2$, then this transition is only possible, by definition, through the transition $s \xrightarrow{a} s'$ for some $s'$ such that $s_x' \overset{\mathrm{(\ref{eq:a-transition-2})}}{=} (t_{a,s'},s_\uparrow)$. We see that the pair $\langle s',s_x' \rangle \in \mathcal{R}$ and is covered in case~\ref{prf:case-4}.
        \item $a \in \Sigma_2 \land \overline{a} \not\in \Sigma_1$. This is a symmetric case where $s \xrightarrow{a} s'$ and $s_x' \overset{\mathrm{(\ref{eq:a-transition-2})}}{=} (s_\downarrow,t_{a,s'})$.  We see that the pair $\langle s',s_x' \rangle \in \mathcal{R}$ and is covered in case~\ref{prf:case-5}.
        \item $a \in \Sigma_1 \land \overline{a} \in \Sigma_2$, then the only transition possible is the $\tau$ transition of~(\ref{eq:a-transition-1}). Then either $s_x' = (\Lbag s_j, \epsilon \Rbag,\Lbag s_i, t_{s_j} \Rbag)$ or $s'_x = (\Lbag s_i, t_{s_j} \Rbag, \Lbag s_j, \epsilon \Rbag)$ where $j \in \{\uparrow,\downarrow\}$ and $j \neq i$. We see that in both possible values of $s'_x$, the pair $\langle s, s_x' \rangle \in \mathcal{R}$ and is covered in cases~\ref{prf:case-2}~and~\ref{prf:case-3}.
        \end{enumerate}
    \end{enumerate}
    \item \label{prf:case-2} Consider a pair $\langle s, s_x \rangle = \langle s, (\Lbag s_i, t_{s_j}\Rbag, \Lbag s_j, \epsilon \Rbag) \rangle$.
    \begin{enumerate}
        \item Assume $s \xrightarrow{a} s'$. Then $s_x \xrightarrow{\tau}_x (\Lbag s_j, \epsilon \Rbag, \Lbag s_j, \epsilon \Rbag)$, and we covered the pair $\langle s, (\Lbag s_j, \epsilon \Rbag, \Lbag s_j, \epsilon \Rbag) \rangle$ in case~\ref{prf:case-1}.
        \item Assume $s_x \xrightarrow{a}_x s'_x$. The only possible transition in {$\rightarrow_x$} is if $a$ is a $\tau$ action consuming the queue message $t_{s_j}$ then $s'_x = (\Lbag s_j, \epsilon \Rbag, \Lbag s_j, \epsilon \Rbag)$ and we covered the pair $\langle s, (\Lbag s_j, \epsilon \Rbag, \Lbag s_j, \epsilon \Rbag) \rangle$ in case~\ref{prf:case-1}.
    \end{enumerate}
    \item \label{prf:case-3} Consider a pair $\langle s, s_x \rangle = \langle s, (\Lbag s_j, \epsilon \Rbag, \Lbag s_i, t_{s_j}\Rbag) \rangle$. Symmetric to case~\ref{prf:case-2}.
    \item \label{prf:case-4} Consider a pair $\langle s, s_x \rangle = \langle s, (t_{a,s_i}, \Lbag r_i, \epsilon \Rbag ) \rangle$ such that $r \xrightarrow{a} s$. 
    \begin{enumerate}
        \item Assume $s \xrightarrow{b} s'$. Then $s_x \xrightarrow{\tau}_x (\Lbag s_i, \epsilon \Rbag, \Lbag r_i, t_{s_i} \Rbag)$.
        \item Assume $s_x \xrightarrow{b}_x s'_x$. The only possible transition in {$\rightarrow_x$} is if $b$ is a $\tau$ action resulting from the synchronisation of the two transitions $t_{a,s_i} \xrightarrow{t_{s_i}}_i \Lbag s_i, \epsilon \Rbag$ and $\Lbag r_i, \epsilon \Rbag \xrightarrow{\overline{t_{s_i}}}_j \Lbag r_i, t_{s_i} \Rbag$. Then, in the product, $s_x \xrightarrow{\tau}_x (\Lbag s_i, \epsilon \Rbag, \Lbag r_i, t_{s_i} \Rbag)$; and the pair $\langle s, (\Lbag s_i, \epsilon \Rbag, \Lbag r_i, t_{s_i} \Rbag) \rangle \in \mathcal{R}$ and is covered in case~\ref{prf:case-6}.
    \end{enumerate}
    \item \label{prf:case-5} Consider a pair $\langle s, s_x \rangle = \langle s, (\Lbag r_i, \epsilon \Rbag, t_{a,s_i}) \rangle$. Symmetric to case~\cref{prf:case-4}.
    \item \label{prf:case-6} Consider a pair $\langle s, s_x \rangle = \langle s, (\Lbag s_i, \epsilon \Rbag, \Lbag r_i, t_{s_i} \Rbag) \rangle$ such that $r \xrightarrow{a} s$.
    \begin{enumerate}
        \item Assume $s \xrightarrow{b} s'$. Then because of the queue-consuming transition $ \Lbag r_i, t_{s_i} \Rbag \xrightarrow{\tau}_j \Lbag s_i, \epsilon \Rbag$, then $s_x \xrightarrow{\tau}_x (\Lbag s_i, \epsilon \Rbag, \Lbag s_i, \epsilon \Rbag)$.\\ The pair $\langle s, (\Lbag s_i, \epsilon \Rbag,\Lbag s_i, \epsilon \Rbag) \rangle$ is covered in case~\ref{prf:case-1}.
        \item Assume $s_x \xrightarrow{b}_x s_x'$, then there are two possible values for $b$:
        \begin{enumerate}
            \item Action $b$ is a queue-consuming $\tau$, then $s_x \xrightarrow{\tau}_x (\Lbag s_i, \epsilon \Rbag, \Lbag s_i, \epsilon \Rbag)$; and the pair \\ $\langle s, (\Lbag s_i, \epsilon \Rbag,\Lbag s_i, \epsilon \Rbag) \rangle$ is covered in case~\ref{prf:case-1}.
            \item $b \in \Sigma_i$, then $s_x \xrightarrow{b}_x (t_{b,s'_i}, \Lbag r_i, t_{s_i} \Rbag)$ such that $s \xrightarrow{b} s'$; and the pair $\langle s', (t_{b,s'_i}, \Lbag r_i, t_{s_i} \Rbag) \rangle \in \mathcal{R}$.
        \end{enumerate}
    \end{enumerate}
    \item Consider a pair $\langle s, s_x \rangle = \langle s, \Lbag r_i, t_{s_i} \Rbag, \Lbag s_i, \epsilon \Rbag) \rangle$. Symmetric to case~\ref{prf:case-6}.
    \item \label{prf:case-8} Consider a pair $\langle s, s_x \rangle = \langle s, (t_{a,s_i}, \Lbag \Lbag p_i, \epsilon \Rbag, t_{r_i} \Rbag) \rangle$ such that $p \xrightarrow{b} r \xrightarrow{a} s$. Then $s_x \xrightarrow{\tau} (t_{a,s_i},\Lbag r_i, \epsilon \Rbag)$; and the pair $\langle s, (t_{a,s_i},\Lbag r_i, \epsilon \Rbag) \rangle$ is covered in case~\ref{prf:case-4}.
    \item Consider a pair $\langle s, s_x \rangle = \langle s, (\Lbag p_i, t_{r_i} \Rbag, t_{a,s_i}) \rangle$ such that $p \xrightarrow{b} r \xrightarrow{a} s$. This is symmetric to case~\ref{prf:case-8}.
\end{enumerate}
\end{proof}

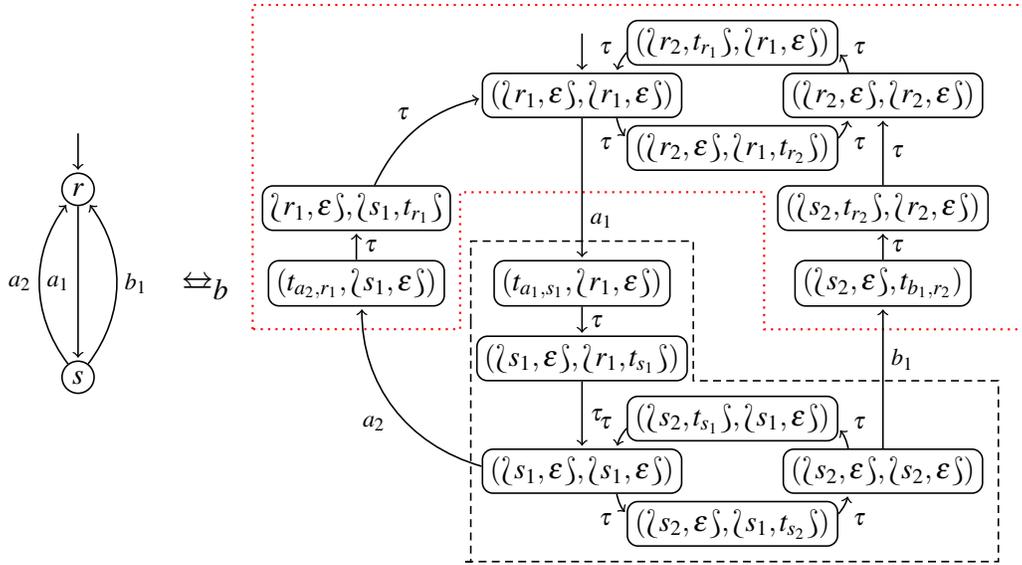
\begin{figure}
    \centering
\begin{tikzpicture}[->,shorten >=1pt,auto,node distance=2cm,on grid,semithick, inner sep=1mm, minimum size=0pt,bend angle=20, 
    line/.style={
      draw,thick,
      -latex',
      shorten >=2pt
    }]
\node [scale=1.4] (decomp) {$\leftrightarroweq_{b}$};
\node[initial,initial where=above,initial text=,state, minimum size=0pt] (r) [above left= 1.25cm and 1.7cm of decomp]                    {$r$};
\node[state,minimum size=0pt]         (s) [below left= 1.25cm and 1.7cm of decomp] {$s$};
\path [every node/.style={font=\footnotesize}, bend angle=37]
(r) edge              node [anchor=east] {$a_1$} (s)
(s) edge [bend left]  node {$a_2$} (r)
    edge [bend right] node [anchor=west] {$b_1$} (r)
;
\node[draw,rectangle,rounded corners] (tr1) [right=of decomp] {$(t_{a_2,r_\uparrow},\Lbag s_\uparrow, \epsilon \Rbag)$};
\node[draw,rectangle,rounded corners] (ts1) [right= 3cm of tr1] {$(t_{a_1,s_\uparrow},\Lbag r_\uparrow, \epsilon \Rbag)$};
\node[initial,initial where=above,initial text=,draw,rectangle,rounded corners, minimum size=0pt] (r1)                   [above =2.5cm of ts1] {$(\Lbag r_\uparrow, \epsilon \Rbag,\Lbag r_\uparrow, \epsilon \Rbag)$};
\node[draw,rectangle,rounded corners, minimum size=0pt]         (r2) [right= 4cm of r1] {$(\Lbag r_\downarrow, \epsilon \Rbag,\Lbag r_\downarrow, \epsilon \Rbag)$};
\node[draw,rectangle,rounded corners, minimum size=0pt]         (r2r1) [above right= 0.7cm and 2cm of r1] {$(\Lbag r_\downarrow, t_{r_\uparrow} \Rbag, \Lbag r_\uparrow, \epsilon \Rbag)$};
\node[draw,rectangle,rounded corners, minimum size=0pt]         (r1r2) [below right= 0.7cm and 2cm of r1] {$(\Lbag r_\downarrow, \epsilon \Rbag, \Lbag r_\uparrow, t_{r_\downarrow}\Rbag)$};
\node[draw,rectangle,rounded corners, minimum size=0pt]         (s1) [below = 2.5cm of ts1] {$(\Lbag s_\uparrow, \epsilon \Rbag,\Lbag s_\uparrow, \epsilon \Rbag)$};
\node[draw,rectangle,rounded corners, minimum size=0pt]         (s2) [right= 4cm of s1] {$(\Lbag s_\downarrow, \epsilon \Rbag,\Lbag s_\downarrow, \epsilon \Rbag)$};
\node[draw,rectangle,rounded corners, minimum size=0pt]         (s2s1) [above right= 0.7cm and 2cm of s1] {$(\Lbag s_\downarrow, t_{s_\uparrow} \Rbag, \Lbag s_\uparrow, \epsilon \Rbag)$};
\node[draw,rectangle,rounded corners, minimum size=0pt]         (s1s2) [below right= 0.7cm and 2cm of s1] {$(\Lbag s_\downarrow, \epsilon \Rbag,\Lbag s_\uparrow, t_{s_\downarrow}\Rbag)$};
\node[draw,rectangle,rounded corners] (tr2) at ($(r2)!0.5!(s2)$) {$(\Lbag s_\downarrow, \epsilon \Rbag,t_{b_1,r_\downarrow})$};
\node[draw,rectangle,rounded corners] (tr2t) [above = 1cm of tr2] {$(\Lbag s_\downarrow, t_{r_\downarrow}\Rbag, \Lbag r_\downarrow, \epsilon \Rbag)$};
\node[draw,rectangle,rounded corners] (ts1t) [below = 1cm of ts1] {$(\Lbag s_\uparrow, \epsilon \Rbag, \Lbag r_\uparrow, t_{s_\uparrow} \Rbag)$};
\node[draw,rectangle,rounded corners] (tr1t) [above = 1cm of tr1] {$\Lbag r_\uparrow, \epsilon \Rbag, \Lbag s_\uparrow, t_{r_\uparrow} \Rbag$};
\draw [red,thick,dotted] 
  ($ (tr1.south west) + (-0.2,-0.3) $) -- 
  ($ (tr1.north west)+(-0.2,0.2) $) |- 
  ($ (r1.north west)+(-0.2,0.9) $) -- 
  ($ (r2.north east)+(0.6,0.9) $) |- 
  ($ (tr2.south east)+(0.18,-0.3) $) -- 
  ($ (tr2.south west)+(-0.4,-0.3) $) -- 
  ($ (tr2.north west)+(-0.4,.9) $) -- 
  ($ (tr1.north east)+(0.2,.9) $) -- 
  ($ (tr1.south east)+(0.2,-0.3) $) -- 
  cycle;
\draw [densely dashed] 
  ($ (ts1.south west)+(-0.3,-0.2) $) -- 
  ($ (ts1.north west)+(-0.3,0.25) $) -- 
  ($ (ts1.north east)+(0.3,0.25) $) -- 
  ($ (ts1.south east)+(0.3,-0.2) $) |- 
  ($ (s2.north east)+(0.2,0.9) $) -- 
  ($ (s2.south east)+(0.2,-0.9) $) -- 
  ($ (s1.south west)+(-0.25,-0.9) $) -| 
  cycle;

\path [every node/.style={font=\footnotesize}, bend angle=14]
(r1) edge  [pos=0.72] node {$\boldsymbol{a_1}$} (ts1)
     edge  [bend right, swap] node {$\tau$} (r1r2)
(r1r2) edge  [bend right, swap] node {$\tau$} (r2)
(r2) edge  [bend right, swap] node {$\tau$} (r2r1)
(r2r1) edge  [bend right, swap] node {$\tau$} (r1)
(ts1) edge node {$\tau$} (ts1t)
(ts1t) edge node {$\tau$} (s1)
(s1) edge [bend left=37]  node {$\boldsymbol{a_2}$} (tr1)
     edge [bend right, swap] node {$\tau$} (s1s2)
(s1s2) edge [bend right, swap] node {$\tau$} (s2)
(tr1) edge [swap] node {$\tau$} (tr1t)
(tr1t) edge [bend left=25] node {$\tau$} (r1)
(s2) edge node [swap, pos=0.63] {$\boldsymbol{b_1}$} (tr2)
     edge [bend right, swap] node {$\tau$}  (s2s1)
(s2s1) edge [bend right, swap] node {$\tau$} (s1)
(tr2) edge [swap] node {$\tau$} (tr2t)
(tr2t) edge [swap] node {$\tau$} (r2)
;
\end{tikzpicture}
    \caption{Showing branching bisimulation following \cref{fig:decompExample2}.}
    \label{fig:async-bb}
\end{figure}

\begin{corollary}
\label{thm:bb-async}
It follows from \cref{lem:bbproof-async} that there is a universal way of decomposing an LTS $M$
using a general asynchronous decomposition operator (\cref{dfn:decomposition_general}) such that $M$ is branching-bisimilar to the synchronous product of its decomposition.
\end{corollary}

\section[Proof that no decomposition operation maintains divergence-preserving branching bisimulation]{Proof that no decomposition operation maintains $\leftrightarroweq_{db}$}
\label{sec:proofdpbb}

In this section, we prove that there is no way of decomposing an LTS such that it is divergence-preserving branching-bisimilar to the synchronous product of its decomposition.

We define confluence based on~\cite{groote1996confluence}.

\begin{dfn}[Confluence]
An LTS $(S, \Sigma_1 \cup \Sigma_2, \rightarrow, s_0)$ is called confluent over $\Sigma_1$ and $\Sigma_2$ iff for all states $s,s_a,s_b \in S$ and for all $a \in \Sigma_1$ and $b \in \Sigma_2$, if $s \xrightarrow{a} s_a$ and $s \xrightarrow{b} s_b$, then there is a state $s_c$ such that $s_b \xrightarrow{a} s_c$ and $s_a \xrightarrow{b} s_c$.
\end{dfn}

\begin{lem}
\label{lem:product-confluence}
Any LTS $M$ that is the synchronous product (\cref{dfn:syncproduct}) of two LTSs $M_1$ and $M_2$ whose action sets are $\Sigma_1$ and $\Sigma_2$ respectively, is confluent over two sets $\Sigma_1 \setminus \overline{\Sigma_2}$ and $\Sigma_2 \setminus \overline{\Sigma_1}$.
\end{lem}

\begin{proof}
Consider the synchronous product $(S_1 \times S_2,\Sigma_x, \rightarrow_x, (q_0, r_0))$ from \cref{dfn:syncproduct} and some actions $a \in \Sigma_1 \setminus \overline{\Sigma_2}$ and $b \in \Sigma_2 \setminus \overline{\Sigma_1}$. Then $a \neq \overline{b}$. Consider some states $s, s' \in S_1$, $t, t' \in S_2$, and $s_a \in S_1 \times S_2$. We know that if $(s,t) \xrightarrow{a} s_a$ then that is due to a transition $s \xrightarrow{a} s'$  and that makes $s_a = (s',t)$, and that given a transition $t \xrightarrow{b} t'$, then a transition $(s',t) \xrightarrow{b} (s',t')$ is possible. Similarly, if $(s,t) \xrightarrow{b} (s,t')$ then $(s,t') \xrightarrow{a} (s',t')$. Therefore, the defined synchronous product is confluent.
\end{proof}

\cref{fig:dpbb-example} (centre) shows a simple LTS $P$.
Concretely, it is defined as $(\{p,r,s\}, \Sigma_1 \cup \Sigma_2, {\rightarrow}, p)$ with alphabets $\Sigma_1=\{a\}$ and $\Sigma_2=\{b\}$ and transitions $p \xrightarrow{a} r$ and $p \xrightarrow{b} s$. In the following lemma and theorem, we prove that no way of decomposing $P$ maintains divergence-preserving branching bisimulation.

\begin{lem}
\label{lem:aorb}
Given the LTS $P$ (\cref{fig:dpbb-example}, centre) with action set $\Sigma_1 \cup \Sigma_2$, let $P_1$ and $P_2$ be two LTSs with action sets $\Sigma_{P_1}$ and $\Sigma_{P_2}$ respectively, and with $\Sigma_1 \subseteq \Sigma_{P_1}$ and $\Sigma_2 \subseteq \Sigma_{P_2}$ and $\Sigma_1 \cap \Sigma_2 = \emptyset = \Sigma_1 \cap \Sigma_{P_2} = \Sigma_{P_1} \cap \Sigma_2 $.
Let $P_x$ be the synchronous product $P_1 \times P_2$ by \cref{dfn:syncproduct}. Then $P \not\leftrightarroweq_{db} P_x$. 
\end{lem}

\begin{proof}
We prove this lemma by contradiction. Assume that $P \leftrightarroweq_{db} P_x$, and let $p_x$ be the initial state of $P_x$, then $p \leftrightarroweq_{db} p_x$. As $p$ is not divergent and cannot do a $\tau$-transition, it holds that only finite sequences of $\tau$'s are possible from $p_x$. This can be seen as follows. If $p_x \xrightarrow{\tau} p_1 \xrightarrow{\tau} p_2 \xrightarrow{\tau} \cdots$, then $p \leftrightarroweq_b p_i$ for all $i > 0$. Hence, $p_x$ is divergent. But this is not possible because $p$ is not divergent.
So, $p_x$ takes a finite number of $\tau$ steps to reach some state $p'_x$ where $p'_x \not\xrightarrow{\tau}$.

Since it must be that $p \leftrightarroweq_{db} p'_x$, and since $p \xrightarrow{a} r$ and $p \xrightarrow{b} s$ where $a \in \Sigma_1$ and $b \in \Sigma_2$, then there are two states $r_x$ and $s_x$ such that $p'_x \xrightarrow{a} r_x$ and $p'_x \xrightarrow{b} s_x$, and $r \leftrightarroweq_{db} r_x$ and $s \leftrightarroweq_{db} s_x$. 

Now because $a \in \Sigma_1 \setminus \overline{\Sigma_2}$ and $b \in \Sigma_2 \setminus \overline{\Sigma_1}$, then $P_x$ is confluent over these two sets, then there must exist a state $p''_x$ such that $r_x \xrightarrow{b} p''_x$. However, $r \not\xrightarrow{b}$. Therefore, $r \not\leftrightarroweq_{db} r_x$. Contradiction. Therefore $P \not\leftrightarroweq_{db} P_x$.
\end{proof}

The proof is illustrated in \cref{fig:dpbb-example} showing that divergence-preserving branching bisimulation ($\leftrightarroweq_{db}$) does not hold when decomposing the LTS $P$ due to the confluence property of decompositions. On the other hand (literally the other hand of the same figure), branching bisimulation holds when decomposing LTS $P$. The reason it holds under $\leftrightarroweq_b$, but not under $\leftrightarroweq_{db}$ is that the former admits infinite $\tau$ cycles, i.e. divergence, which, as demonstrated here in right side of the figure, avoids the premise of confluence altogether. \cite{Baeten2013Reactive} provides a similar insight into how divergence maintains branching bisimilarity but breaks divergence-preserving branching bisimilarity. 

\begin{thm}
\label{thm:dpbb}
There is no decomposition operation that maintains divergence-preserving branching bisimulation ($\leftrightarroweq_{db}$) for all LTSs.
\end{thm}
\begin{proof}
We prove this theorem by contradiction. Assume that there is a decomposition operation that maintains $\leftrightarroweq_{db}$ for all LTSs. Then it must do so for any arbitrary LTS $P$. But since \cref{lem:aorb} proves that no LTS maintains $\leftrightarroweq_{db}$ for one such LTS $P$, i.e the one in \cref{fig:dpbb-example}, then there is no decomposition operation that maintains $\leftrightarroweq_{db}$ for all LTSs. 
\end{proof}

\begin{figure}
    \centering
\begin{tikzpicture}[->,shorten >=1pt,auto,node distance=2cm,on grid,semithick, inner sep=1mm, minimum size=0pt,bend angle=20]
\node [scale=1.4] (dpb) {$\not\leftrightarroweq_{db}$};
\node[initial,initial where=above,initial text=,state, minimum size=0pt] (p) [above left= .25cm and 1.5cm of dpb] {$p$};
\node [scale=1.4] (bb) [below left= .25cm and 1.5cm of p] {$\leftrightarroweq_{b}$};
\node[state,minimum size=0pt] (r) [below left= 1.25cm and 0.7cm of p] {$r$};
\node[state,minimum size=0pt] (s) [below right= 1.25cm and 0.7cm of p] {$s$};
\path [every node/.style={font=\footnotesize}, bend angle=37]
(p) edge node [anchor=south east] {$a$} (r)
(p) edge node {$b$} (s)
;

\node[initial,initial where=above,initial text=,state, minimum size=0pt] (p_x_1) [above left= .25cm and 1.3cm of bb]  {$p_x$};
\node[state,minimum size=0pt] (p1_1) [left= 1cm of p_x_1] {$ $};
\node[state,minimum size=0pt] (r_x_1) [left= 1.25cm of p1_1] {$ $};
\node[state,minimum size=0pt] (s_x_1) [below= 1.25cm of p1_1] {$s_x$};
\node[state,minimum size=0pt] (pp_1)  [below= 1.25cm of r_x_1] {$r_x$};
\path [every node/.style={font=\footnotesize}, bend angle=57]
(p_x_1) edge node [swap,pos=0.35] {$\tau^*$} (p1_1)
(p1_1)  edge [bend left=27] node {$\tau$} (r_x_1)
(p1_1)  edge node {$b$} (s_x_1)
(r_x_1) edge node {$a$} (pp_1)
(r_x_1) edge [bend left=27] node {$\tau$} (p1_1)
;

\node[initial,initial where=above,initial text=,state, minimum size=0pt] (p_x_2) [above right= .25cm and 1.3cm of dpb]  {$p_x$};
\node[state,minimum size=0pt] (p1_2) [right= 1cm of p_x_2] {$ $};
\node[state,minimum size=0pt] (r_x_2) [right= 1.25cm of p1_2] {$r_x$};
\node[state,minimum size=0pt] (s_x_2) [below = 1.25cm of p1_2] {$s_x$};
\node[state,minimum size=0pt] (pp_2) [below = 1.25cm of r_x_2] {$p'$};
\path [every node/.style={font=\footnotesize}, bend angle=37]
(p_x_2) edge node  {$\tau^*$} (p1_2)
(p1_2) edge node  {$a$} (r_x_2)
(p1_2) edge node {$b$} (s_x_2)
(r_x_2) edge node {$b$} (pp_2)
;

\end{tikzpicture}
    \caption{Illustration for \cref{lem:aorb}.}
    \label{fig:dpbb-example}
\end{figure}

\section{Interpretation}
\label{sec:conclusion}
One way to understand this fundamental result is that if the subsystems of the decomposition must communicate, then there is no escape from introducing divergence in order to maintain equivalence over any and all decompositions of LTSs. 

With respect to automata learning, this result implies that, unless one values divergency, it is not possible to make any assumption about the distribution of components based on the information that is learned. If one can observe divergencies while learning behaviour, it might be possible to say something about the internal structure of a system though this will be highly non trivial to accomplish. 

And with respect to software, our result says that it is always possible to distribute a piece of software over different components if one allows divergent behaviour.  Otherwise, such a distribution is not possible and based on our proof. One can see that this already applies to very simple behaviours. 

Furthermore, divergence, in an industrial context, is undesired due to the requirement of fairness, i.e., one subsystem seizing unfair control over the total behaviour of the system through infinite looping. This means that if some decomposition is found to maintain fairness, then that is guaranteed not to be the case universally over all contexts and all LTSs. 

\newpage

\bibliographystyle{eptcs}
\bibliography{decompBib}
\end{document}